\documentclass[conference]{IEEEtran}
\IEEEoverridecommandlockouts
\usepackage{cite}
\usepackage{amsmath,amssymb,amsfonts}
\usepackage{graphicx}
\usepackage{textcomp}
\usepackage{xcolor}
\usepackage{fdsymbol}
\usepackage{url}
\usepackage{xcolor}
\usepackage{graphicx}
\usepackage[top=0.75in,left=0.64in,bottom=1in,right=0.62in,textwidth=7.26in]{geometry}
\usepackage[ruled,linesnumbered]{algorithm2e}

\newtheorem{theorem}{Theorem}
\newtheorem{definition}{Definition}
\newtheorem{proof}{Proof}

\newtheorem{assumption}{Assumption}

\def\BibTeX{{\rm B\kern-.05em{\sc i\kern-.025em b}\kern-.08em
    T\kern-.1667em\lower.7ex\hbox{E}\kern-.125emX}}

\begin{document}

\title{Matching Using Sufficient Dimension Reduction for Heterogeneity Causal Effect Estimation\\
\thanks{\#These authors contributed equally to this work}
}

\author{
    \IEEEauthorblockN{Haoran~Zhao$^{1\#}$, Yinghao~Zhang$^{1\#}$, Debo~Cheng$^{2,3*}$, Chen Li$^{4}$, Zaiwen Feng$^{5,6,7,8,9,10*}$}
    \IEEEauthorblockA{$^1$ College of Science, Huazhong Agricultural University, Wuhan, 430070, China}
    \IEEEauthorblockA{$^2$ College of Computer Science and Engineering, Guangxi Normal University, Guilin, 541000, China}
	\IEEEauthorblockA{$^3$ STEM, University of South Australia, Australia}
	\IEEEauthorblockA{$^4$ Department of Computer Science and Systems Engineering, Kyushu Institute of Technology, Iizuka 820-8502, Japan}
	\IEEEauthorblockA{$^5$ College of Informatics, Huazhong Agricultural University, Wuhan 430070, China}
 \IEEEauthorblockA{$^6$ Hubei Key Laboratory of Agricultural Bioinformatics, Huazhong Agricultural University, Wuhan, 430070, China}
   \IEEEauthorblockA{$^7$ Hubei Hongshan Laboratory, Huazhong Agricultural University, Wuhan 430070, China}
     \IEEEauthorblockA{$^8$ National Key Laboratory of Crop Genetic Improvement, Huazhong Agricultural University, Wuhan, 430070, China}
    \IEEEauthorblockA{$^9$ State Key Laboratory of Hybrid Rice, Wuhan University, 299 Bayi Rd, Wuhan, 430070, China}
    \IEEEauthorblockA{$^{10}$ Macro Agricultural Research Institute, Huazhong Agricultural University, Wuhan 430070, China}

		*Correspondence: Debo~Cheng (chengdb2016@gmail.com) and Zaiwen Feng (Zaiwen.Feng@mail.hzau.edu.cn)
		
}

\maketitle
	
		\begin{abstract}
		Causal inference plays an important role in understanding the underlying mechanisation of the data generation process across various domains. It is challenging to estimate the average causal effect and individual causal effects from observational data with high-dimensional covariates due to the curse of dimension and the problem of data sufficiency. The existing matching methods can not effectively estimate individual causal effect or solve the problem of dimension curse in causal inference. To address this challenge, in this work, we prove that the reduced set by sufficient dimension reduction (SDR) is a balance score for confounding adjustment.  Under the theorem, we propose to use an SDR method to obtain a reduced representation set of the original covariates and then the reduced set is used for the matching method. In detail, a non-parametric model is used to learn such a reduced set and to avoid model specification errors. The experimental results on real-world datasets show that the proposed method outperforms the compared matching methods. Moreover, we conduct an experiment analysis and the results demonstrate that the reduced representation is enough to balance the imbalance between the treatment group and control group individuals. 
	\end{abstract}
	
	\begin{IEEEkeywords}
		Matching, Sufficient Dimension Reduction, Causal Inference, Individual Causal Effect.
	\end{IEEEkeywords}
	
	\section{Introduction}
	
	In recent decades, causal inference has gained increasing attention across many areas, such as economic~\cite{imbens2015causal}, health~\cite{connors1996outcomes}, statistic~\cite{ghosh2020sufficient} and computer science~\cite{pearl2009causality}. Causal effect estimation plays an important role in revealing the causal strength between two factors in causal inference. A randomised control experiment is regarded as the standard gold for identifying such causal strength. However, the randomised control experiment is impractical due to the cost, time, or ethic~\cite{imbens2015causal,pearl2009causality}. Therefore, estimating causal effects from observational data is an important alternative in causal inference~\cite{nabi2017semiparametric,li2018mapreduce,ghosh2020sufficient,cheng2020causal,cheng2022toward}. 
	
	Confounding bias is the main challenge in the causal effect estimations from observational data~\cite{cattaneo2010efficient,ghosh2020sufficient}. The commonly used approach is confounding adjustment to remove the confounding bias when estimating the causal effects from observational data~\cite{imbens2015causal}. One of the most popular confounding adjustments is the matching method~\cite{stuart2010matching}. The core idea of the matching method is to balance the distribution of covariates between the treatment group and the control group. In detail, the first step of matching is to seek units from the control group with similar covariates to those in the treatment group for constructing a matching pair, and then the causal effect can be calculated based on the matched data. Matching can be performed by selecting different functions of the covariates and selecting different matching algorithms such as Mahalanobis distance matching, full matching, nearest neighbour matching, and genetic matching~\cite{stuart2010matching}. 
	
	The most widely used function is propensity scores which are regarded as the probability of a unit receiving a treatment~\cite{rubin1979using,rosenbaum1983central}. In general, the matching method involves transforming multi-dimensional covariates into scalars by using propensity scores so as to overcome the difficulties in matching based on original covariates~\cite{rubin1973matching,rubin2007design}. In fact, the true propensity score is unknown in real applications, and thus the estimation of the propensity score from data is required. However, the errors of the model or model misspecified inevitably occur in the calculation of the propensity score from data \cite{hahn1998role,ghosh2020sufficient,greenewald2021high}. Hence, the propensity score is not a good solution for the matching method~\cite{luo2017estimating}.
	
	Recently, sufficient dimension reduction has been successfully utilised for estimating causal effects from observational data~\cite{ghosh2020sufficient,cheng2022sufficient,nabi2017semiparametric}. For instance, Luo et al. proposed a sufficient dimension reduction matching for causal effect estimation~\cite{luo2019matching}. The proposed matching method utilises a sufficient dimension reduction method to reduce the original covariates into reduced-dimensional covariates that retain the advantages of both the original covariates and the propensity score under mild assumptions. The advantage of the reduced-dimensional covariates is the asymptotic stability and is superior to that of the estimates obtained by using propensity scores. However, this method is to obtain two sets of reduced reduced-dimensional covariates by using sufficient dimension reduction on sub-datasets, i.e. the treated samples and the control samples, but not discovering an adjustment set over the whole data. 
	
	In a data, if the number of samples is not big enough, the performance of the proposed method will be decreased significantly. Nabi et al.~\cite{nabi2017semiparametric} proposed a semi-parametric causal sufficient dimension reduction method to deal with multiple treatments. Cheng et al.~\cite{cheng2022sufficient} proposed a CESD matching method for the average causal effect estimation by using a kernel dimension reduction method~\cite{fukumizu2004dimensionality} to reduce the covariates relative to the treatment variable, but not for heterogeneity causal effect estimation.

	In this work, we propose a novel \underline{M}atching method based on \underline{I}nverse \underline{R}egression \underline{E}stimator (referred to as MIRE method) for average and heterogeneity causal effect estimation from observational data. In detail, our MIRE method utilises a sufficient dimension reduction method, i.e. the inverse regression estimator, to learn reduced-dimensional covariates relative to the outcome variable over the whole data. Then, the MIRE method utilises the reduced-dimensional covariates to conduct a matching process for imputing the unobserved outcomes (a.k.a counterfactual outcomes)~\cite{imbens2015causal,ghosh2020sufficient,cheng2022sufficient}. Our experimental analysis shows that the reduced-dimensional covariates are well-balanced which is why MIRE addresses the confounding bias very well.

	To summarise, our work makes the following contributions.
	\begin{itemize}
		\item We tackle the problem of estimating heterogeneity causal effect estimation from observational data with sufficient dimension reduction. 
		\item We propose a novel matching method based on the inverse regression estimator, MIRE, for causal effect estimation from observational data.
		\item Extensive experiments demonstrate that our proposed matching method is more effective in terms of the causal effect estimation from observational data.
	\end{itemize}

	\section{background}
	
	\subsection{Potential Outcome Model}
	
	Let $T$ to be a binary treatment that includes $T_{i} = 0$ (a controlled unit) and $T_{i} = 1$ (a treated unit), where $i$ represents an unit. The set of units which are not assigned to a certain treatment is the control group, and the set of units which are assigned to a certain treatment is the treatment group. The set $X$ is a set of pretreatment covariates, i.e. the covariates are unchanged before and after that the treatment and outcome variables are observed~\cite{imbens2015causal,cheng2022sufficient}. The potential outcomes $Y_{i}$ is defined as the outcome of unit $i$, i.e. $Y_{i}(0)$ and $Y_{i}(1)$ are the potential outcomes of unit $i$ unassigned and assigned to a treatment. Note that both potential outcomes for a unit $i$ can not be observed at the same time. It belongs to a fundamental challenging problem in causal inference~\cite{rubin1974estimating,imbens2015causal}. Under the potential outcome model, the individual causal effect and average causal effects can be defined as follows.
	
	\begin{equation}
		ITE=Y_{i}(1)-Y_{i}(0),
		\label{(1)}  
	\end{equation}
	\begin{equation}
		ATE=E[Y(1)-Y(0)].
		\label{(2)}  
	\end{equation}
	
	The following assumptions are usually required when we utilise the potential outcome model to estimate the causal effects from observational data.
	
	\begin{assumption}[Stable Unit Treatment Value Assumption (SUTVA)~\cite{imbens2015causal}]
		The potential outcome of a unit is not affected by whether or not other units are treated. That means, for each unit, its potential outcomes only rely on the treatment $T$, and there is no different form or version of each treatment level.
	\end{assumption}
	
	\begin{assumption}[Unconfoundedness~\cite{imbens2015causal}]
		Conditioning on the set of covariate $X$, the treatment $T$ was independent of the potential outcomes $Y$, formally $Y(1), Y(0) \Vbar T|X$.
	\end{assumption}
	
	This assumption shows that all units with the set of covariates $X$ are randomly assigned to treatment.
	
	\begin{assumption}[Overlap~\cite{imbens2015causal}] 
		For each unit, it has a non-zero probability to being treated or control when given the set of covariate $X$, i.e.
		$0<P(T=t|X)<1, t=0, 1$.
	\end{assumption}
	
	This assumption shows that there is a probability that each unit is assigned to a treatment $t$. Note that the assumptions of Unconfoundedness and Overlap are usually called ``the ignorability assumption''. The ignorability assumption is not testable directly from data since the counterfactual outcomes are unmeasured~\cite{imbens2015causal}. Consequently, the set of covariates $X$ consists of all relevant and irrelevant covariates in terms of estimating the causal effect of $T$ on $Y$. Hence, it is necessary to discover an adjustment set $Z$ from $X$ to accurately estimate the causal effect of $T$ on $Y$. The propensity score always plays a role in causal effect estimation from observational data and is defined below.

	\begin{definition}[Propensity score~\cite{imbens2015causal}]
		The propensity score is defined as the conditional probability of a unit being assigned to a treatment conditioning on the set of covariates $X$.
		\begin{equation}
			e(x)=P(T=1|X)
			\label{(3)}
		\end{equation}
	\end{definition}
	
	The balance score denoted as $b(x)$ is proposed by Rosenbaum and Rubin~\cite{rosenbaum1983central} that allows a class of functions to model the covariate $X$.  In practice, the balance score should be satisfied the unconfoundedness assumption as well (Lemma 12.2 in~\cite{imbens2015causal}). 
	\begin{equation}
		Y(1), Y(0)\Vbar T|b(x)
		\label{(4)}
	\end{equation}
	
	The balance scores contain the propensity scores and the original covariates, and others. Ignorability assumption also suggests that units with the same or approximately equal balance scores have the same distribution of covariates.
	
	\subsection{Matching Method}
	The matching method is to identify the units in the control group with a similar distribution of covariates to the units in the treatment group, so that the potential outcomes of the units in the control group are used to impute the missing potential outcomes of units in the treatment group. The essential idea of the matching method is to simulate the process of randomised control experiments. Thus the matched units can be regarded as the counterfactual outcome of units~\cite{rubin1973matching,rubin1974estimating,rubin2007design}. In a matching method, the potential outcome of the $i$-th unit can be obtained according to the formulas \ref{(5)} and \ref{(6)}.
	
	\begin{equation}
		\label{(5)}
		\hat{Y}_{i}(1)=\left\{\begin{array}{ll}
			Y_{i},  & \text { if } t=1 \\
			\frac{1}{|\ell(i)|} \sum_{{j} \in \ell(i)} Y_{j}, & \text { if } t=0\\
		\end{array}\right. 
	\end{equation}
    
	\begin{equation}
		\label{(6)}
		\hat{Y}_{i}(0)=\left\{\begin{array}{ll}
			\frac{1}{|\ell(i)|} \sum_{{j} \in \ell(i)} Y_{j}, & \text { if } t=1\\
			Y_{i},  & \text { if } t=0 \\
		\end{array}\right.
	\end{equation}
    
	\noindent where $\ell$ denotes that the sample set matches the $i$-th unit from the control or treatment group.
	
	In general, the distribution of covariates in the matched data is more similar between two groups than before matching. Therefore, the matched results can be used to calculate the average causal effect so as to reduce the influence of confounding factors.

	\section{The Proposed MIRE Method}
	In this section, we first prove that $\psi(X) = X^{T} \beta$ by sufficient dimension reduction (SDR) is a balance score for addressing the confounding bias when estimating the causal effects from observational data. Then we introduce our proposed MIRE method (Matching based on the inverse regression estimator) for causal effect estimation.
	
	\subsection{A Central DRS is a Sufficient Balance Score}
	Sufficient dimension reduction (SDR) is a dimension reduction method which is widely used for data processing~\cite{cook1996graphics,cook2009regression,fukumizu2004dimensionality}. For a response variable $Y$ and a set of covariates $X$, SDR is to learn a function $\psi(X)$ such that the original covariates $X$ can be reduced into a subspace $X^{T}\beta \in \mathbb{R}^{p\times k}$ with $k\ll p$. In this work, we assume that the function $\psi(X) = X^{T}\beta$ is existed and can be expressed as follows. 
	\begin{equation}
		\label{(7)}
		Y\Vbar X|X^{T} \beta
	\end{equation}
	
	\noindent where the column subspace of $\beta$ is called the dimension reduction space (DRS). Hence, it is important to learn the subspace $\beta$ in the SDR method.  It is worth noting that  $X^{T} \beta$ obtained by using an SDR method to reduce the dimension of the covariate $X$ and can be viewed as a function of $\psi(X)$.  When a subspace $S_{Y|Z}$ is the intersection of all other dimension reduction subspaces, the subspace $S_{Y|Z}$ is well-known as the central DRS~\cite{cook1996graphics,fukumizu2004dimensionality}. The central DRS has the smallest dimension and unique dimension-reduction subspace~\cite{connors1996effectiveness}. Thus, in our work, we would like to learn the central DRS $S_{Y|Z}$. 
	
	\begin{theorem}
		\label{theo:001}
		Given an observational data $O$ that contains the treatment $T$, the outcome $Y$, and the set of the pretreatment variables $X$. Suppose that the central subspace $S_{Y|Z}$ is existing and with $r$-dimensional, where $r\ll p$. Then there is an arbitrary basis matrix $\beta \in \mathbb{R}^{p\times k}$ such that $X^{T} \beta$ satisfies $Y(1), Y(0)\Vbar T|X^{T} \beta$ and is a balancing score.
	\end{theorem}
	\begin{proof}
		Under the pretreatment variable assumption, $X$ has not had a descendant node of either $W$ or $Y$. Under the ignorability assumption, the causal effect of $T$ on $Y$ can be calculated unbiasedly based on confounding adjustment or adjusting for a balance score. The existence of a central subspace $S_{Y|Z}$ ensures that there is an arbitrary basis matrix $\beta \in \mathbb{R}^{p\times k}$ such that $Y\Vbar X| X^{T} \beta$ holds according to the invariant property of central subspace~\cite{cook1996graphics,ghosh2020sufficient}. Mathematically, $X^{T} \beta\cong X$ holds. Moreover, the unconfoundedness assumption, i.e. $Y(1), Y(0)\Vbar T|X$ holds. So replacing $X$ in $Y(1),Y(0)\Vbar T|X$ with $X^{T} \beta$, we have $Y(1), Y(0)\Vbar T|X^{T} \beta$. Therefore, $X^{T} \beta$ is a balance score according to the invariant property of central subspace and $Y(1), Y(0)\Vbar T|X^{T} \beta$.	
	\end{proof}
	
	Based on Theorem~\ref{theo:001}, $X^{T} \beta$ is a balance score for unbiased causal effect estimation from observational data. That means, it is sufficient to use  $X^{T} \beta$ as balance scores in the matching method for unbiased causal effect estimations. It is worth noting that another advantage of SDR is able to reduce the dimension of original covariates while retaining the important information.
	
	In this work, we adopt the inverse regression estimator (IRE) method for learning the central DRS since IRE belonging to the inverse regression (IR) method family is an optimal method with the highest asymptotic efficiency~\cite{cook2005sufficient}.
	
	When $Y$ is the continuous value, $Y$ is discretised with its range divided into $h$ slices based on the previous notation~\cite{cook2005sufficient}. The central subspace $\beta$ can be obtained by calculating the following formula.
	\begin{equation}
		\beta_{\xi_{y}} = {\textstyle \sum_{y=1}^{h}}Span(\xi _{y}) 
	\end{equation}
	\noindent where $\xi_{y}=\Sigma ^{-1}(E(X|Y=y)-E(X))$.
	
	Furthermore, we assume that the linearity condition for estimating central subspace., i.e. $E(Z|P_{S_{Y|Z}}Z)=P_{S_{Y|Z} }Z$ is induced, based on which the central subspace is linked to the inverse regression of $Z$ on $Y$.

	\subsection{implementation of MIRE}
	
	In this study, we use the inverse regression estimator (IRE)~\cite{cook2005sufficient} to estimate the central DRS for our MIRE method. The IRE method is to estimate the DRS by minimising the objective function~\ref{(8)}. First, the following equation is used to calculate the quadratic discrepancy for the IRE method.
	
	\begin{equation}
		\begin{aligned}
			\label{(8)}
			F_{k}^{IRE} (S, C)=&(vec(\hat{\zeta})-vec(SC)) ^{T}\hat{\Gamma} _{\hat{\zeta}} ^{-1} \\
			&(vec(\hat{\zeta})-vec(SC))
		\end{aligned}
	\end{equation}
	\noindent where $\hat{\Gamma}_{\hat{\zeta}} ^{-1}  $ is a nonsingular covariance matrix. The columns of $S\in \mathbb{R}^{p\times k}$ represent a basis for Span($\xi$), and $C\in \mathbb{R}^{k\times (h-1)}$ represents the coordinates of $\xi$ relative to $S$. $\hat{\zeta  }$ satisfies $\hat{\zeta  }\equiv \beta \gamma D_{f}A $, where $\beta\in \mathbb{R}^{p\times k}$ is a basis of DRS. $A$ is a nonstochastic matrix satisfies $A^{T}A=I_{h-1}  $ and $A^{T}1_{h}=0 $. $\gamma$ is a vector such that $\xi=\beta \gamma$. $D_{f}$ is a diagonal matrix with the elements of the vector $f$ on the diagonal, where $\hat{\mathrm{f}}=(\hat{f}_{1}, ..., \hat{f}_{h})$.  $vec(\cdot)$ denotes the operator that constructs a vector from a matrix by stacking its columns and can be formalised as follows.
	
	\begin{equation}
		\begin{aligned}
			vec(C)=&[(I_{h-1}\otimes S^{T}) V_{n}(I_{h-1}\otimes S)  ]^{-1}\times \\
			&(I_{h-1}\otimes S^{T}) V_{n}vec(\hat{\xi } )
			\label{(9)}
		\end{aligned}
	\end{equation}
	\noindent where $\otimes$ is the Kronecker product, which is an operation on two matrices of arbitrary size resulting in a block matrix. $V_{n}$ is a positive-definite matrix and is equal to a consistent estimate $\hat{\Gamma} _{\hat{\xi } } ^{-1}$ of $\Gamma _{\hat{\zeta}} ^{-1}$.
	
	\begin{equation}
		\begin{aligned}
			\hat{S} _{d} = Q_{S_{(-d)} } [Q_{S_{(-d)} }(c_{d}^{T}\otimes I_{p}  )Q_{S_{(-d)} }]^{-} \times \\Q_{S_{(-d)} }(c_{d}^{T}\otimes I_{p})V_{n}\alpha _{d}  
			\label{(10)}
		\end{aligned}
	\end{equation}
\noindent where $\alpha _{d} = vec(\hat{\zeta} -S_{(-d)} C_{(-d)})$, $C_{d}$ is the $d$-th row of $C$, $C_{(-d)}$ consists of all but the $d$th row of $C$, and  $Q_{S_{(-d)} }$ projects onto the orthogonal complement of Span($S_{(-d)}$) in the usual inner product.
	
	The matching process is performed according to the above-mentioned matching steps, and the distance is measured by using the following Mahalanobis distance.
	\begin{equation}
	D_{ij}=(\psi(X)_{i}-\psi(X)_{j}) {}' \Sigma ^{-1} (\psi(X)_{i}-\psi(X)_{j})
	\label{eq:000332}
	\end{equation}

\noindent where $\Sigma$ is the covariance matrix of all the units. In practice, the observed data are used to calculate the covariance matrix.
	
	The MIRE method is as shown in Algorithm~\ref{alg:algorithm1}.
	
	\begin{algorithm}[t]
		\caption{MIRE}
		\label{alg:algorithm1}
		\KwIn{Observational data $O$ includes the set of covariates $X$, treatment $T$ and outcome $Y$.}
		\KwOut{paired datasets}  
		\BlankLine
		Initialise $S\gets (s_{1},...,s_{d})$ randomly
		
	    Calculate the least squares coefficient for $B$ using equation (\ref{(9)}) with fixed $S$  
		
		Assign $e \gets F_{d}(S, C)$ and $iter\gets 0$
		
		\For{$d\leftarrow 1$ \KwTo $k$}{
			$S = (s_{1},...,s_{k})$
			
			Find a new $s_{d} $ by equation \ref{(10)}
			
			$\hat{s}_{d}\gets\frac{\hat{s}_{d}}{\left \| \hat{s}_{d}  \right \|} $
			
			Update
			$S\gets (s_{1},...,\hat{s}_{d},...,s_{_{k} })$
			
			$C\gets \arg_{c^{\ast } } \min F_{k}(S, C^{\ast } ) $
			
			$e \gets F_{k}(S, C); iter\gets iter+1  $
			
			\If {$e$ no longer decreases}{
				
				$\tilde{S} \gets S$
			} 
			
			$\tilde{S} \gets S$
			
		}
		Calculate an ordered basis for Span($\widetilde{B}$) (i.e. $\hat{b}_{1}$, $\hat{b}_{2}$)
		
		$X^{\ast} = X(\hat{b }_{1}, \hat{b }_{2}  )$
		
		\For{$i\leftarrow 1$ \KwTo $n{_{1}}$}{
			\For{$j\leftarrow 1$ \KwTo $n{_{2}}$}{
				
				$D_{ij}=(X^{\ast}_{i} -X^{\ast}_{j})' \Sigma ^{-1} (X^{\ast}_{i} - X^{\ast}_{j})$}
			Find $j$ with the smallest $D_{ij}$
			
			Pair the units $i, j$.
		}
	\end{algorithm}
	
	After matching, the standardised difference in the mean of the covariate balance. Based on the matched data, the causal effect can be obtained. The average causal effects can be estimated by using the formula~\ref{(8)}. The individual causal effects of the $i$-th unit can be estimated according to the equation $\hat{Y_{i}}(1)-\hat{Y_{i}}(0)$. 

\section{Experiment}
\label{sec:exp}
	It is very difficult to evaluate the proposed causal effect estimation methods by using the real-world datasets because we cannot know the counterfactual outcomes in the real-world datasets~\cite{rubin1979using,imbens2015causal}. To evaluate the performance of the proposed MIRE method, we select two semi-synthetic datasets, IHDP~\cite{hill2011bayesian} and TWINS~\cite{almond2005costs}. Both semi-synthetic datasets are widely used for evaluating causal  effect estimation methods. Moreover, one real-world dataset, Jobs~\cite{lalonde1986evaluating} is also used in our experiments since the dataset have the empirical causal effects in the literature.

	To evaluate the performance of the proposed MIRE method, seven commonly used causal effect estimation methods were selected for comparison, including: NNM (Nearest neighbor matching~\cite{rubin1973matching}), PSM (propensity score matching~\cite{rosenbaum1983central}), BART (Bayesian Additive Regression Trees~\cite{hill2011bayesian}), CF (Causal Forest~\cite{athey2019generalized}), SDRM (Sufficient dimension reduction matching~\cite{luo2019matching}), BCF (Bayesian Causal Forest~\cite{hahn2020bayesian}), R-LASSO (R-learner using LASSO Regression~\cite{nie2021quasi}). These methods have been regarded as one of the most efficient causal effect estimation methods as their ability of eliminating confounding bias, i.e. the state-of-the-art method in causal effect estimation method.

	For the experiments on the IHDP, we use Precision in Estimating Heterogeneous Treatment Effects (PEHE) $\mathrm{PEHE}=\frac{1}{N}  {\textstyle \sum_{i=1}^{N}} ((y_{i1}-y_{i0})-(\hat{y}_{i1}-\hat{y}_{i0}))^{2}$ as an evaluation criterion for assessing the heterogeneous causal effects. For experiments on the real dataset Jobs, we estimate the average causal effect on the treated samples (ATT) since the empirical ATT is known~\cite{imai2014covariate}. For experiments on TWINS, we estimated the average causal effect (ATE). In addition, we use root-mean-square error (RMSE) $\mathrm{RMSE}\!=\sqrt{\frac{1}{{~N}} \sum_{{i}=1}^{{N}}\left({y}_{{i}}\!-\hat{{y}}_{{i}}\right)^{2}}$ and standard deviation (SD) as evaluation metrics for assessing the performance of all methods.
	
	\subsection{Estimation of heterogeneous effects based on IHDP}
	The benchmark dataset IHDP in causal inference is from a simulation study conducted in the work~\cite{hill2011bayesian}. The data is from the Infant Health and Development Program (IHDP), a program that began in 1985 to provide high-quality home visiting services to low birth weight preterm infants. The results of the program showed that after treatment (i.e., after receiving the service), there was a significant increase in cognitive test scores in the treatment group compared to the control group at the age of $3$.
	
	A variety of covariates was collected in this study such as child birth weight, head circumference, weeks of prematurity, birth order,  and neonatal health indicators, as well as maternal behaviour during pregnancy and some indicators during delivery. To simulate the imbalance between the treatment and control groups, we discard the non-random portion of the treatment group from the experimental data as suggested by Hill~\cite{hill2011bayesian}, specifically all children of non-white mothers, while leaving the control group intact. The potential outcomes for each unit were  then simulated by creating response surfaces so that true individual causal effects could be concluded. And because the response surface is known, the covariates that generate the response surface can be adjusted to satisfy the Ignorability assumption.
	We follow the response surface B used by Hill~\cite{hill2011bayesian}: 
	$$Y(0) \sim  N(exp((X+W)\beta _{B} ), 1)$$
	$$Y(1) \sim  N(X\beta _{B} -\omega _{B} ,1)$$
	\noindent where $W$ is an offset matrix with the same dimension as X with every value equal to 0.5, $\beta _{B}$ is a vector of regression coefficients (0, 0.1, 0.2, 0.3, 0.4) randomly sampled with probabilities (0.5, 0.125, 0.125, 0.125, 0.125) for the 6 continuous covariates and (0.6, 0.1, 0.1, 0.1, 0.1) for the 18 binary covariates, $\omega_{B}$ is an offset chosen to guarantee that ATT = 4. Therefore, the true individual causal effect can be calculated at this point. 
	
	Figure \ref{fig1} shows the results of 1,000 simulations of this data using different methods. For each method, we calculate the average of the 1,000 results as the PEHE value. From this figure, it can be seen clearly that MIRE has the best performance among all the compared methods, which also shows the effectiveness of the method.  The second good performance of methods is R-LASSO and SDRM. Meanwhile, the widely used method, CF (Causal Forest) also has a good performance but is worse than MIRE, R-LASSO and SDRM. The rest methods, NNM (Nearest neighbour matching), PSM (propensity score matching), BART (Bayesian Additive Regression Trees) and BCF (Bayesian Causal Forest) have worse performance than MIRE, R-LASSO, SDRM and CF.
	
	\begin{figure}[htbp]
		\centerline{\includegraphics[scale=0.6]{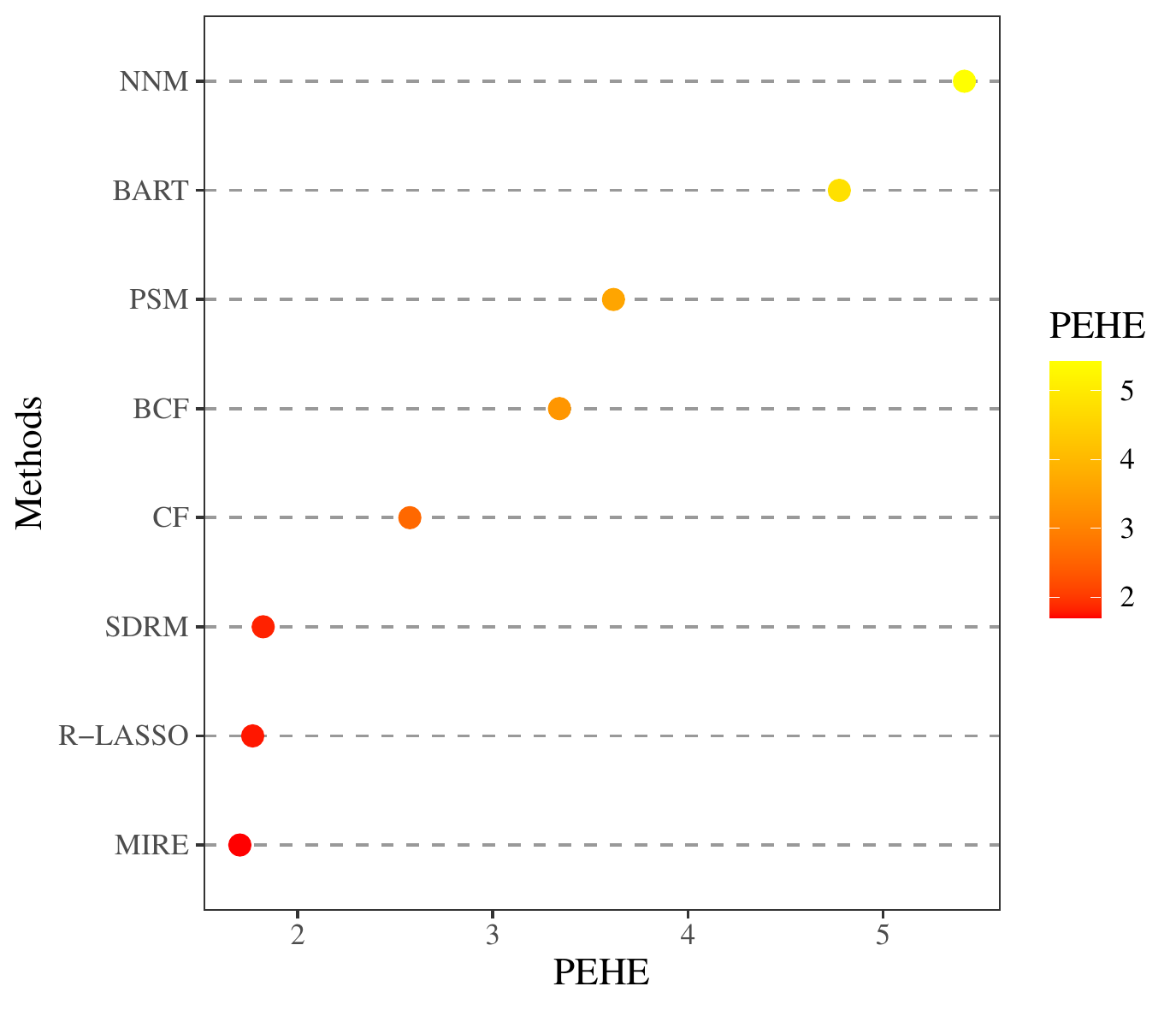}}
		\caption{PEHE values estimated using different methods in the IHDP simulation experiment.}
		\label{fig1}
	\end{figure}
	
	\subsection{Estimation of causal effects}
	\subsubsection{Jobs}
	The Jobs dataset is a classic dataset established by LaLonde in 1986 for causal inference~\cite{lalonde1986evaluating}. The dataset was derived from a temporary employment program aimed to provide work skills to people who are faced with economic hardship or lack job skills. The treatment is whether an individual participates in the program, and the potential outcome is the individual income in the year 1978. Covariates mainly included education, age, race, marital status, income in 1974, and income in 1975.  The average causal effect on the treated samples in the Lalonde dataset was estimated as \$886 with a standard error of \$448, and we used this estimate (\$886) as a criterion to evaluate our method \cite{lalonde1986evaluating,diamond2013genetic,imai2014covariate,cheng2022sufficient}.

	Table \ref{table3} shows the results of estimating the average causal effect on the treated samples on the Jobs dataset using different methods. From table \ref{table3}, it can be seen that MIRE has a good performance for ATT estimation, and the estimated ATT (\$519.09) is close to the criterion (\$886) with a small SD (\$734.93) that is also close to the empirical SD (\$448). The two methods BART and BCF also show a good performance on this dataset. While the CF estimated ATT has a larger difference from the standard value.
	
	 Figure \ref{fig4} shows the scatter plot after dimension reduction of our MIRE on the jobs dataset, which shows the relationship between the reduced variables and the response variable $Y$. From Figure \ref{fig4}, it can be seen that both covariates after dimension reduction are significantly correlated with the response variable $Y$. The result also shows the rationality of using dimension reduction covariates for matching.

	\begin{table}[!ht]
		\centering
		\caption{Estimated average causal effect on the treated samples on Jobs dataset.}
		\begin{tabular}{|c|c|c|c|c|}
			\hline
			\textbf{Methods}  & \textbf{Estimated ATT} & \textbf{RMSE} & \textbf{SD} \\ 
			\hline
			NNM & 198.16 & 683.34 & 1,280.70 \\ \hline
			PSM & 1,933.40 & 1,090.60 & 702.20 \\ \hline
			BART & 931.30 & 2182.02 & 1,032.00 \\ \hline
			CF & 182.24 & 497.65 & 890.61 \\ \hline
			SDRM & 1,740.98 & 1,025.65 & 710.55 \\ \hline
			R-LASSO & 1,271.63 & 385.63 & 825.23 \\ \hline
			BCF & 697.88 & 544.90 & 511.82 \\ \hline
			MIRE & 519.09 & 734.92 & 734.93 \\ \hline
		\end{tabular}
		\label{table3}
	\end{table}
	
	\begin{figure}[htbp]
		\centerline{\includegraphics[scale=0.6]{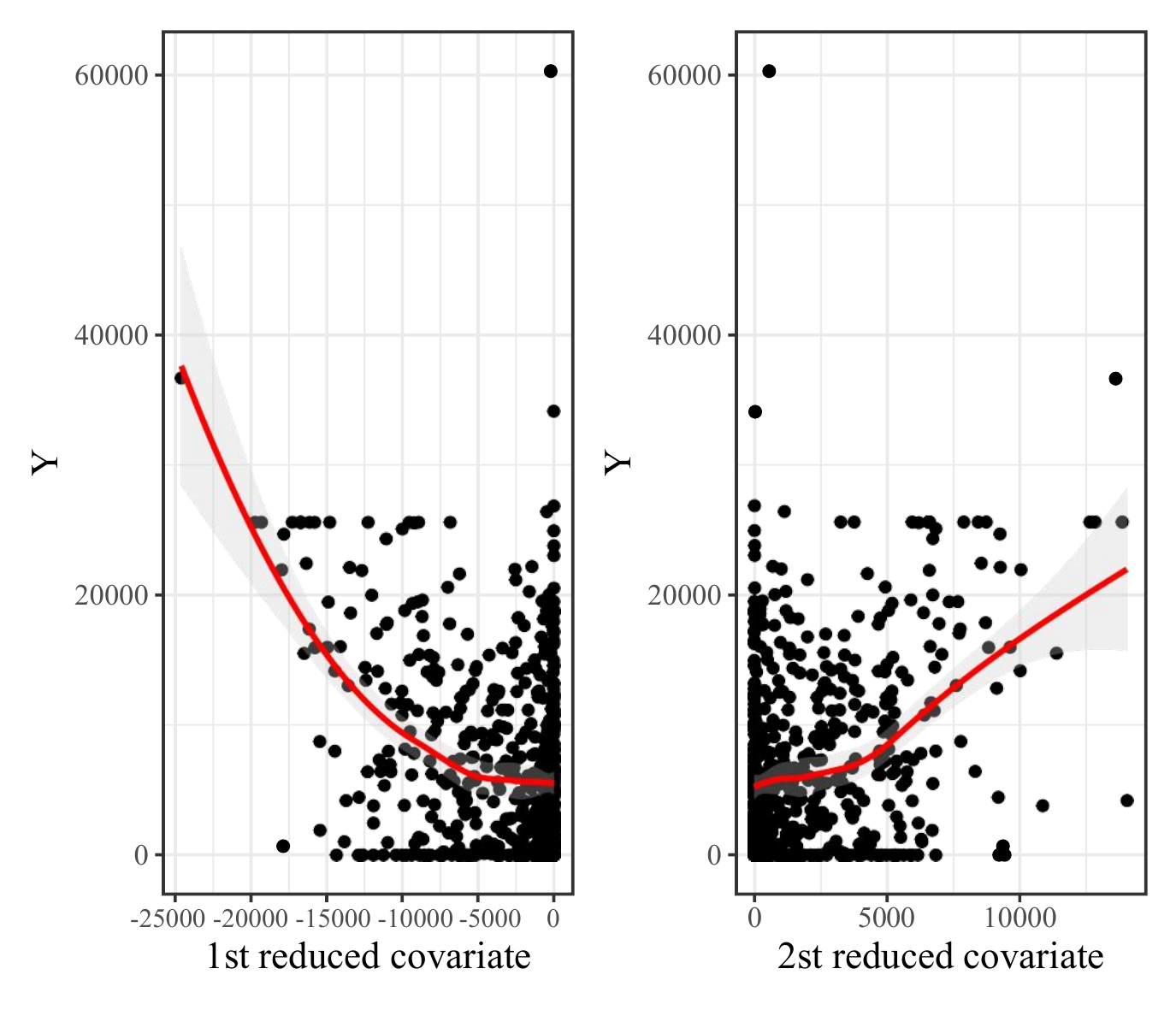}}
		\caption{The relationships between reduced covariates and response variable $Y$ by our MIRE on Jobs dataset.}
		\label{fig4}
	\end{figure}
	
	\subsubsection{TWINS}
	The TWINS benchmark dataset was created based on the data of twins born in the United States between 1983 and 2000, the treatment refers to the heavier weight of the twins at birth~\cite{almond2005costs}. The potential outcome is the mortality in the first year of twin birth. In Louizos et al.’s study~\cite{louizos2017causal}, one of the two twins was selectively hidden, which was equivalent to randomly assigning the treatment, thus making it similar to the data of randomdised experiment. Then those twins with weight less than 2kg were selected to establish a dataset. The dataset included 40 covariates such as parents' education, marital status, race, and mother's condition at the time of delivery. We simulated the presence of confounding factors according to the following formula~\cite{louizos2017causal}:
	$W_{i} |X_{i}\sim  Bern(Sigmiod(w{}'X_{i})+n)$, where $W\sim U(-0.1,0.1)^{40\times 1}, n\sim N(0,0.1)$.
	
	The average causal effect of our established dataset is $-0.025$. Table \ref{table4} shows the results of estimating the average causal effect on the TWINS dataset using different methods. It can be seen from the results shown in table \ref{table4}: MIRE, CF, BCF, PSM both have really well performance. This also shows that in the existence of confounders, MIRE is able to remove the confounding bias as the stat-of-the-art causal effect estimators.
	
	 Figure \ref{fig5} shows the scatter plot after dimension reduction of the covariates on the TWINS dataset, which shows the relationship between the reduced variables and the response variable $Y$. From figure \ref{fig5}, we have that the covariates after dimension reduction are significantly correlated with the response variables. It also confirms that MIRE is effective for estimating causal effects from observational data.

	\begin{table}[htbp]
		\begin{center}
			\caption{Estimated average causal effect  on TWINS dataset.}
			\begin{tabular}{|c|c|c|c|c|}
				\hline
				\textbf{Methods}  & \textbf{Estimated ATE} & \textbf{RMSE} & \textbf{SD} \\ 
				\hline
				NNM & -0.0218 & 0.0029 & 0.0468 \\ \hline
				PSM & -0.0252 & 0.0043 & 0.0502 \\ \hline
				BART & -0.1794 & 0.2793 & 0.3188 \\ \hline
				CF & -0.0252 & 0.0034 & 0.0408 \\ \hline
				SDRM & -0.0230 & 0.0032 & 0.0480 \\ \hline
				R-LASSO & -0.0623 & 0.0221 & 0.0901 \\ \hline
				BCF & -0.0249 & 0.0398 & 0.0639 \\ \hline
				MIRE & -0.0252  & 0.0037 & 0.0002 \\ \hline
			\end{tabular}
			\label{table4}
		\end{center}
	\end{table}
	
	\begin{figure}[htbp]
		\centerline{\includegraphics[scale=0.6]{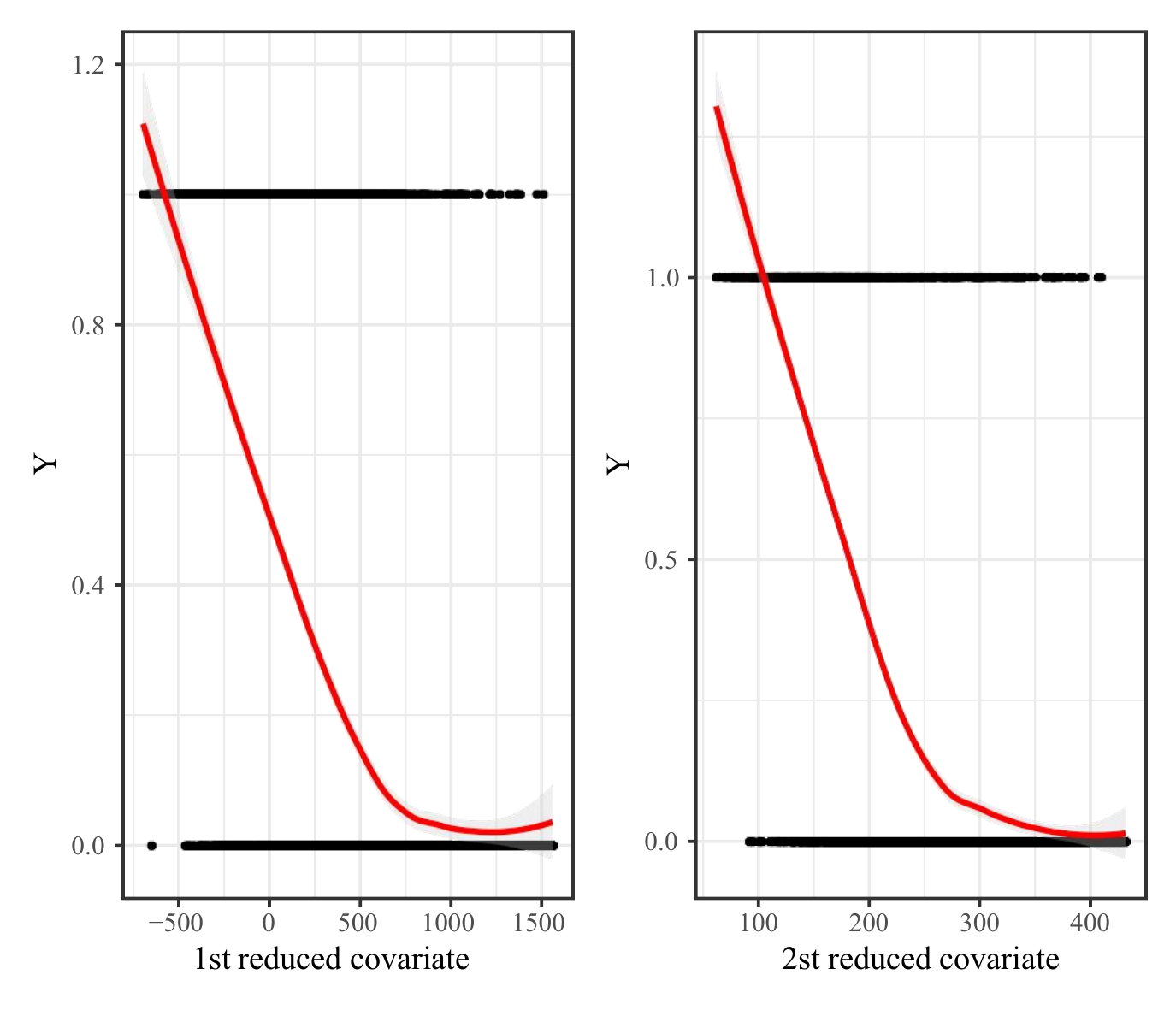}}
		\caption{The relationships between reduced covariates and response variable $Y$ after our MIRE on TWINS dataset.}
		\label{fig5}
	\end{figure}

	\section{Related Work}
	The potential outcome model is widely used in causal inference~\cite{rubin1973matching,rubin1974estimating,rubin2007design,imbens2015causal}. Our proposed MIRE method builds on the potential outcome model with mild assumptions, such as the assumptions of the pretreatment variables and the ignorability assumption. The matching method is one of the most commonly used methods in causal inference~\cite{rubin1973matching,luo2019matching,cheng2022sufficient}. Stuart summarises the principles and steps of the matching method as well as its application in many practical problems~\cite{stuart2010matching}. In the following, we review some works that are closely related to our proposed MIRE method.   
	
	It is difficult to find individuals with multiple dimensions of the covariates between two groups (control and treatment) during the matching process. To solve this problem, Rubin and Rosenbaum~\cite{rubin1973matching} introduced a propensity score, defined as the conditional probability that a unit is assigned to a certain treatment under the covariate condition. In addition, the propensity score has been proven to be a balance score~\cite{imbens2015causal}.  The most commonly used propensity score matching method is to use the propensity score in place of the original covariate in matching progress. In addition, propensity score matching and Mahalanobis distance matching are combined to generate GenMatch, and GenMatch uses a genetic search algorithm to obtain weights so as to complete matching \cite{stuart2010matching}. Luo and Zhu used a sufficient dimension reduction (SDR) method to reduce the dimension of the covariates in the treatment group and control group, and then matched them based on Mahalanobis distance \cite{luo2017estimating}.
	
	The most related work to MIRE is the SDR matching proposed by the works~\cite{luo2019matching,cheng2022sufficient}. Luo and Zhu's work~\cite{luo2019matching} consider reducing the sub-datasets over the treated units and the control units to obtain two of the reduced-dimensional covariates as the balance score for matching. The proposed method maybe suffers from bias since dividing the whole samples into two sub-datasets results in data insufficiency. Cheng et al.~\cite{cheng2022sufficient} aim to estimate the average causal effect from observational data, but not for heterogeneity causal effect estimation. In contrast, our theoretical findings support a data-driven method for heterogeneity causal effect estimation.

	In recent, a large number of deep learning-based methods have been proposed for estimating the causal
	effects from observational data~\cite{shalit2017estimating,yao2018representation,yoon2018ganite,shortreed2017outcome,kallus2019interval}. 
	The main advantage of deep learning-based methods is that the complex nonlinear relationships between variables can be learned by neural networks and the high-dimensional datasets can be addressed very well. Nevertheless, a number of parameters turning are very inefficient, and they do have not good interpretability.
	
	Another line work on causal effect estimation from data with latent confounders~\cite{kallus2019interval,cheng2020causal,cheng2022toward}. When an instrumental variable (IV) is given, the causal effect of $T$ on $Y$ can be calculated unbiasedly from data with
	latent variable too~\cite{athey2019generalized,cheng2022ancestral,cheng2022discovering,hernan2006instruments,martens2006instrumental}. Because IV-based estimators do not rely on the ignorability assumption, they are not directly related to our MIRE method.
	
	\section{Conclusion}
	In this work, we prove that the central DRS by a sufficient dimensional reduction method is a balance score and is sufficient to control for confounding bias in causal effect estimation from observational data.  Our findings provide theoretical support for using the dimension-reduced covariates for matching.  Under the proposed theorem, we propose a data-driven method, i.e. MIRE,  to estimate the causal effects from observational data under mild assumptions. Firstly, MIRE utilises the inverse regression estimator to reduce the dimensions of the original covariates, and then uses the reduced-dimensional covariates for matching. The advantages of our proposed MIRE have been verified through the experiments. First, the results of average causal effect estimation based on the Jobs dataset showed that the estimation results of our method were closer to the criterion value (\$886) recommended in the previous study than those of other matching methods. Second, the results of individual causal effect estimation showed that matching based on dimension-reduced covariates made it easier for individuals to be paired with another group in the matching process. Our method displayed great advantages in estimating individual causal effects over other matching methods. The estimation results based on the IHDP dataset indicated that our method could match more individuals during the matching process. Compared with other causal inference methods, our method also exhibited certain advantages in heterogeneous effect estimation accuracy (expressed as PEHE) and confounding factor control.
	

\section*{Acknowledgment}
This research project was supported in part by  the Major Project of Hubei Hongshan Laboratory under Grant 2022HSZD031, and in part by the Innovation fund of Chinese Marine Defense Technology Innovation Center under Grant JJ-2021-722-04, and in part by the National Natural Science Foundation of China under Grant Nos. 62076041 and 61806027, and in part by the Fundamental Research Funds for the Chinese Central Universities under Grant 2662020XXQD01, 2662022JC004, and in part by the open funds of State Key Laboratory of Hybrid Rice, Wuhan University, and in part by the open funds of the National Key Laboratory of Crop Genetic Improvement under Grant ZK202203, Huzhong Agricultural University.
	
\bibliographystyle{IEEEtran}
\bibliography{HPCCbib}

\end{document}